\PassOptionsToPackage{expansion=false}{microtype}
\documentclass{lmcs}

\usepackage[utf8]{inputenc}
\usepackage[T1]{fontenc}
\usepackage{mathtools}

\newcommand{\R}{\mathbb{R}}
\newcommand{\Z}{\mathbb{Z}}
\newcommand{\Q}{\mathbb{Q}}
\newcommand{\N}{\mathbb{N}}
\newcommand{\Qp}{\mathbb{Q}_{+}}
\newcommand{\FO}{\mathrm{FO}}
\newcommand{\dom}{\mathrm{dom}}
\newcommand{\av}[1]{\vec{#1}}

\title{Extending the Ginsburg--Spanier Theorem to Functions and Mixed Arithmetic}

\author[Finkel]{Alain Finkel\lmcsorcid{0000-0002-6699-1815}}
\address{Université Paris-Saclay, CNRS, ENS Paris-Saclay, LMF, 91190, Gif-sur-Yvette, France}
\email{finkel@ens-paris-saclay.fr}

\author[Leroux]{J\'er\^ome Leroux\lmcsorcid{0000-0002-4932-6164}}
\address{LaBRI, Universit\'e de Bordeaux, CNRS, LaBRI, 33405, Talence, France}
\email{leroux@labri.fr}

\begin{document}
\maketitle

\begin{abstract}
We study sets and functions definable in three classical additive theories:
Presburger arithmetic (first-order logic over the integers with addition and order),
the real additive theory (first-order logic over the reals with addition and order),
and the mixed additive theory (first-order logic over both reals and integers with
addition and order).

The Ginsburg--Spanier theorem characterizes the sets definable in Presburger arithmetic
as exactly the semi-linear sets.
We extend this characterization in two directions.

First, we show that the functions definable in Presburger arithmetic are exactly the
piecewise linear functions, and that the same holds for the real additive theory.
The proofs are direct algebraic arguments using only a stability lemma and the
Ginsburg--Spanier theorem.

Second, we introduce \emph{semi-polinear sets} as the analogue of semi-linear sets
for the mixed additive theory, and prove that the mixed-definable sets are exactly
the semi-polinear sets.
We further show that the functions definable in the mixed additive theory are exactly
the \emph{piecewise-simple} functions, a new class of functions that are linear
separately in the integer part and in the fractional part of the argument,
but with potentially different linear coefficients for each part.

These algebraic characterizations unify the three theories in a single geometric
framework.
The proofs are purely algebraic, using only linear algebra and the Ginsburg--Spanier
theorem, without reference to automata, quantifier elimination, or machines.
We also correct an error in a prior characterization of mixed-definable sets.
\end{abstract}

\section{Introduction}

\paragraph{Context.}
Presburger arithmetic $\FO(\Z,+,\leq)$~\cite{Pre29} is the decidable first-order theory
of the integers with addition and order.
The Ginsburg--Spanier theorem~\cite{GS66} characterizes its definable sets as exactly
the semi-linear sets, i.e., finite unions of linear sets.
The real additive theory $\FO(\R,+,\leq)$ is decidable via Fourier--Motzkin elimination,
with definable sets being exactly the finite unions of polyhedral convex sets~\cite{Wei88}.
The mixed additive theory $\FO(\R,\Z,+,\leq)$ was shown decidable by
Weispfenning~\cite{Wei99} via quantifier elimination, and later by Boigelot,
Jodogne and Wolper~\cite{BJW05} via an alternative approach using weak B\"uchi automata.

Concerning Presburger \emph{functions}, prior work studied
computability by flat (without nested loops) counter automata~\cite{MR67,Che76,Gur85}
or by reversal-bounded counter automata~\cite{GI79},
closure under composition~\cite{IL81}, and
evaluation in linear space and quadratic time~\cite{CI83}.

\paragraph{Motivation.}
The three additive theories $\FO(\Z,+,\leq)$, $\FO(\R,+,\leq)$, and $\FO(\R,\Z,+,\leq)$
are well-studied logics with applications across verification, model checking, and
constraint solving.
A natural question is: what are the \emph{functions} definable in these logics,
described in purely geometric terms?
Prior characterizations of Presburger-definable functions all relied on operational or
logical descriptions --- counter automata~\cite{MR67,Che76,GI79}, closure under
composition~\cite{IL81}, or logical normal forms --- and the same holds for the mixed
theory via quantifier elimination~\cite{Wei99} or B\"uchi automata~\cite{BJW05}.
Surprisingly, a direct geometric answer was missing from the literature.

\begin{table}[h]
\centering
\renewcommand{\arraystretch}{1.3}
\resizebox{\textwidth}{!}{%
\begin{tabular}{|l|c|c|c|}
\hline
 & $\FO(\Z,+,\leq)$ & $\FO(\R,+,\leq)$ & $\FO(\R,\Z,+,\leq)$ \\
\hline
\textbf{Decidability} & \cite{Pre29}, 2-EXPSPACE~\cite{FR74,Ber80}
                      & via Fourier--Motzkin QE
                      & \cite{Wei99,BJW05} \\
\hline
\textbf{Logic / QE}   & Cooper; singly-exp.~\cite{HKMZ24}
                      & Fourier--Motzkin
                      & extended logic~\cite{Wei99} \\
\hline
\textbf{Automata}     & counter automata~\cite{MR67,Che76,GI79}
                      & ---
                      & weak B\"uchi aut.~\cite{BJW05} \\
\hline
\textbf{Sets}         & semi-linear~\cite{GS66}
                      & polyhedral convex~\cite{Wei88}
                      & \textbf{semi-polinear} \\
\hline
\textbf{Functions}    & \textbf{piecewise linear} & \textbf{piecewise linear} & \textbf{piecewise-simple} \\
\hline
\end{tabular}%
}
\caption{Decidable additive theories: logical, automata-theoretic, and algebraic
characterizations. \textbf{Bold} entries are results of the present paper.
``$---$'' indicates that no automata characterization of $\FO(\R,+,\leq)$-definable
sets is known; this appears to be open.}
\label{tab:summary}
\end{table}

\paragraph{Related work.}
The Ginsburg--Spanier theorem~\cite{GS66} is the starting point.
Weispfenning~\cite{Wei99} describes mixed-linear sets but with an error corrected
here (Remark~\ref{rem:wei99}).
The piecewise linear characterization of $\FO(\Z,+,\leq)$-definable functions is
implicitly present in the literature (applying~\cite{GS66} to $G_f$) but was
never stated explicitly as a theorem with a direct algebraic proof.

\paragraph{Contributions.}
\textbf{(1) Presburger functions are piecewise linear.}
A function is definable in $\FO(\Z,+,\leq)$ if and only if it is piecewise linear
on a Presburger partition of its domain (Theorem~\ref{thm:integer}).
The proof uses only a stability argument and the Ginsburg--Spanier theorem.

\textbf{(2) Extension of the Ginsburg--Spanier theorem to the mixed theory.}
We introduce \emph{semi-polinear sets} as the $\FO(\R,\Z,+,\leq)$ analogue of
semi-linear sets (Theorem~\ref{thm:mixed-sets}), and show that functions definable in
$\FO(\R,\Z,+,\leq)$ are exactly the piecewise-simple functions
(Theorem~\ref{thm:mixed-functions}).
We correct an error in Weispfenning~\cite{Wei99} regarding the structure of
mixed-linear sets of $\R$ (Remark~\ref{rem:wei99}).

\paragraph{On the depth of the results.}
The proofs rely on standard techniques from linear algebra and the Ginsburg--Spanier
theorem, and may appear elementary.
We argue that the results are nonetheless of real interest.
First, \emph{surprisingly}, they are \emph{new}: despite the maturity of the field, no paper had stated
the piecewise-linear or piecewise-simple characterizations as explicit theorems with
direct algebraic proofs.
Second, the proofs are \emph{self-contained and short}, making the results easy to
verify and teach.
Third, they \emph{unify} three theories in a single geometric framework, as
Table~\ref{tab:summary} illustrates: the algebraic row was the only one missing.
Fourth, we \emph{correct an error} in the only prior work on mixed-linear
sets~\cite{Wei99}.
Fifth, the notions of piecewise-linear and piecewise-simple functions are defined
\emph{without reference to logic or automata}, making them useful independently
in verification, synthesis, and constraint solving.

\paragraph{Organization.}
Section~\ref{sec:prelim} introduces notations and the three logics of interest.
Section~\ref{sec:real} characterizes definable functions in the real additive theory.
Section~\ref{sec:integer} characterizes definable functions in the integer additive
theory, recalling the Ginsburg--Spanier theorem.
Section~\ref{sec:mixed} introduces semi-polinear sets and proves the extension of
the Ginsburg--Spanier theorem to the mixed case, then characterizes mixed definable
functions as piecewise-simple functions.

\section{Preliminaries}
\label{sec:prelim}

\subsection{Notations}

We denote by $\R$, $\Q$, $\Qp$, $\Z$, $\N$ the set of real numbers, the set of
rational values, the set of non-negative rational values, the set of integers, and the set
of non-negative integers.
The $i$-th component of a vector $\av{x}$ is denoted by $\av{x}[i]$.
The \emph{integer-part function} $[\cdot] : \R \to \Z$ is defined by $[x] = \max\{k \in \Z \mid k \leq x\}$;
for $\av{a} \in \R^n$, $[\av{a}]$ denotes the component-wise integer part.
Functions $f : X \to Y$ are \emph{partially defined} over \emph{definition domains}
$\dom(f) \subseteq X$.
The \emph{graph} of $f$ is the set
$G_f = \{(x,y) \in X \times Y \mid x \in \dom(f) \wedge y = f(x)\}$.
The restriction of $f : X \to Y$ to a subset $D \subseteq \dom(f)$ is denoted by
$f|_D : X \to Y$.

\begin{lemma}
\label{lem:decomp}
Decompositions of $G_f$ into finite unions $G_f = G_1 \cup \cdots \cup G_k$ correspond to
decompositions of $\dom(f)$ into finite unions $\dom(f) = D_1 \cup \cdots \cup D_k$ where
$G_i$ is the graph of $f_i$ with $f_i = f|_{D_i}$.
\end{lemma}
\begin{proof}
Given $G_f = G_1 \cup \cdots \cup G_k$, set $D_i = \{\av{x} \mid \exists \av{y} : (\av{x},\av{y}) \in G_i\}$.
Since $G_f$ is a function graph, each $G_i \subseteq G_f$ is also a function graph
(a point $(\av{x},\av{y}_1) \in G_i$ and $(\av{x},\av{y}_2) \in G_f$ force $\av{y}_1 = \av{y}_2$),
and $G_i$ is the graph of $f|_{D_i}$.
\end{proof}

\subsection{The Three Logics}

The first-order logics $\FO(\Z,+,\leq)$, $\FO(\R,+,\leq)$, and $\FO(\R,\Z,+,\leq)$ are
respectively called the \emph{integer additive theory}, the \emph{real additive theory},
and the \emph{mixed additive theory}.
Let $\mathcal{L}$ be one of these theories and $D \subseteq \R^n$.
A finite partition $\Pi = \{X_1,\ldots,X_k\}$ of $D$ is said \emph{definable in
$\mathcal{L}$} if $X_i$ is definable in $\mathcal{L}$ for any $i$.

A function $f : \R^n \to \R^q$ is said \emph{definable in $\mathcal{L}$} if there exists a
formula $\psi(\av{x},\av{y})$ in $\mathcal{L}$ encoding its graph $G_f$.
Observe that deciding if $\psi(\av{x},\av{y})$ encodes the graph of a function reduces
to the non-satisfiability of:
\[
\psi(\av{x},\av{y}_1) \wedge \psi(\av{x},\av{y}_2) \wedge \av{y}_1 \neq \av{y}_2.
\]
As the logic $\FO(\R,\Z,+,\leq)$ is decidable~\cite{Wei99}, this property is decidable.

\subsection{Piecewise Linear and Stable Functions}

Let $f : \R^n \to \R^q$.
We say that $f$ is \emph{$\Q$-linear} (or simply \emph{linear}) if there exist a matrix
$M \in \Q^{q \times n}$ and a vector $\av{v} \in \Q^q$ such that $f(\av{a}) = M\av{a} +
\av{v}$ for any $\av{a} \in \dom(f)$.
We call such functions
\emph{linear} even though they are sometimes called \emph{affine} in other contexts.
We say that $f$ is \emph{piecewise $\Q$-linear} (or simply \emph{piecewise linear}) if
there exists a finite partition $\Pi$ of $\dom(f)$ such that $f|_X$ is $\Q$-linear for any
$X \in \Pi$; if $\Pi$ is definable in $\mathcal{L}$, we say that $f$ is
\emph{$\mathcal{L}$-piecewise linear}.

We say that $f$ is \emph{stable by addition} (or simply \emph{stable}) if
$\av{a}_1 + \av{a}_2 \in \dom(f)$ and
$f(\av{a}_1 + \av{a}_2) = f(\av{a}_1) + f(\av{a}_2)$
for any $\av{a}_1, \av{a}_2 \in \dom(f)$.

\section{Stable Functions and the Real Additive Theory}
\label{sec:real}

This section establishes two results: that stable functions $f : \R^n \to \R^q$ with $G_f \subseteq \Q^n \times \Q^q$ are
linear (Proposition~\ref{prop:stable}), and that functions definable in $\FO(\R,+,\leq)$
are exactly the piecewise linear ones (Theorem~\ref{thm:real}).

We first recall some elements of linear algebra.
A $\Q$-\emph{vector space} (or simply a \emph{vector space}) of $\Q^n$ is a set
$V \subseteq \Q^n$ such that (1) $\av{0} \in V$, (2) $\av{v}_1 + \av{v}_2 \in V$ for any
$\av{v}_1, \av{v}_2 \in V$, and (3) $\lambda\av{v}$ is in $V$ for any $\lambda \in \Q$ and
any $\av{v} \in V$.
A \emph{basis} of a vector space $V$ is a finite (potentially empty) sequence
$(\av{v}_1,\ldots,\av{v}_d)$ of $d$ vectors in $V$ such that for any $\av{v} \in V$ there
exists a unique sequence $(\lambda_1,\ldots,\lambda_d)$ of elements in $\Q$ satisfying
$\av{v} = \sum_{i=1}^d \lambda_i \av{v}_i$.

\begin{proposition}
\label{prop:stable}
For any stable function $f : \R^n \to \R^q$ with $G_f \subseteq \Q^n \times \Q^q$,
there exists a matrix $M \in \Q^{q \times n}$ such that $f(\av{a}) = M\av{a}$ for any
$\av{a} \in \dom(f)$.
\end{proposition}

\begin{proof}
As $G_f \subseteq \Q^n \times \Q^q$, we deduce that $\dom(f) \subseteq \Q^n$.
Let us consider the vector space $V$ generated by $\dom(f)$.
There exists a basis $(\av{a}_1,\ldots,\av{a}_d)$ of $V$ such that
$\av{a}_i \in \dom(f)$ for any $i$.
As $G_f \subseteq \Q^n \times \Q^q$ we deduce that $f(\av{a}_i) \in \Q^q$ for any $i$.
As $(\av{a}_1,\ldots,\av{a}_d)$ is a basis of $V$ and
$(f(\av{a}_1),\ldots,f(\av{a}_d))$ is a sequence of vectors in $\Q^q$,
there exists a matrix $M \in \Q^{q \times n}$ such that $M\av{a}_i = f(\av{a}_i)$
for any $i$.

Let us prove that $f(\av{a}) = M\av{a}$ for any $\av{a} \in \dom(f)$.
Since the case $d = 0$ is immediate we assume $d \geq 1$.
Let $\av{a} \in \dom(f)$.
As $\dom(f) \subseteq V$ and $(\av{a}_1,\ldots,\av{a}_d)$ is a basis of $V$,
there exists $(\lambda_1,\ldots,\lambda_d)$ in $\Q^d$ such that
$\av{a} = \sum_{i=1}^d \lambda_i \av{a}_i$.
Let $r \geq 1$ be large enough so that $r\lambda_i \in \Z$ for any $i$.
There exist $\mu_i^+, \mu_i^- \in \N$ such that $r\lambda_i = \mu_i^+ - \mu_i^-$.
We deduce:
\[
r\av{a} + \sum_{i=1}^d \mu_i^- \av{a}_i = \sum_{i=1}^d \mu_i^+ \av{a}_i.
\]
The image by $f$ of the two previous vectors can be developed thanks to the stability of
$f$ by addition:
\[
rf(\av{a}) + \sum_{i=1}^d \mu_i^- f(\av{a}_i) = \sum_{i=1}^d \mu_i^+ f(\av{a}_i).
\]
By replacing $\mu_i^+ - \mu_i^-$ by $r\lambda_i$ in the previous equality,
we get $f(\av{a}) = M\av{a}$.
\end{proof}

\subsection{Polyhedral Convex Sets}

A $\Q$-\emph{polyhedral convex set} of $\R^n$ (or simply a \emph{polyhedral convex set})
is a set $C \subseteq \R^n$ such that there exists a finite (potentially empty) sequence
$(\av{\alpha}_j, \#_j, \beta_j)_{1 \leq j \leq k}$ of tuples in
$\Q^n \times \{<,\leq,=,\geq,>\} \times \Q$ satisfying:
\[
C = \Bigl\{ \av{c} \in \R^n \;\Big|\; \forall j\; \sum_{i=1}^n \av{\alpha}_j[i]\,\av{c}[i]
\;\#_j\; \beta_j \Bigr\}.
\]
The \emph{Fourier--Motzkin elimination} shows that the following set $\exists_i C$ is
effectively polyhedral convex for any polyhedral convex set $C \subseteq \R^n$ and
$1 \leq i \leq n$:
\[
\exists_i C = \{(c_1,\ldots,c_{i-1},c_{i+1},\ldots,c_n) \in \R^{n-1} \mid
\exists c_i \in \R\;(c_1,\ldots,c_n) \in C\}.
\]
In particular $\FO(\R,+,\leq)$ admits a quantifier elimination algorithm and we deduce
that a set is definable in $\FO(\R,+,\leq)$ if and only if it is equal to a finite union
of polyhedral convex sets.

\begin{lemma}
\label{lem:nonempty-rational}
Non-empty polyhedral convex sets of $\R^n$ contain vectors in $\Q^n$.
\end{lemma}

\begin{proof}
The proof is by induction over $n$.
For $n = 1$: a non-empty polyhedral convex set $C \subseteq \R$ is an interval with bounds
in $\Q \cup \{-\infty,+\infty\}$.
If $C$ contains a unique value then this value is in $\Q$.
Otherwise, $C$ contains two values $l < h$, so $]l,h[ \subseteq C$, and by density of
$\Q$ in $\R$ there exists $c \in \Q$ with $l < c < h$, hence $c \in C$.
For the induction step, given a non-empty polyhedral convex set $C' \subseteq \R^{n+1}$,
let $I = \exists_{1,\ldots,n} C'$.
Fourier--Motzkin shows $I$ is polyhedral convex and non-empty, so there exists
$c_{n+1} \in I \cap \Q$.
The set $C = \exists_{n+1}\{\av{x} \in C' \mid \av{x}[n+1] = c_{n+1}\}$ is polyhedral
convex and non-empty, so by induction hypothesis there exists
$(c_1,\ldots,c_n) \in C \cap \Q^n$.
Then $(c_1,\ldots,c_{n+1}) \in C' \cap \Q^{n+1}$.
\end{proof}

\begin{proposition}
\label{prop:dense}
$C \cap \Q^n$ is dense in $C$ for any polyhedral convex set $C$ of $\R^n$.
\end{proposition}

\begin{proof}
Let $\av{c} \in C$.
As $\Q$ is dense in $\R$, there exist sequences $(\av{l}_k)_{k \geq 0}$ and
$(\av{h}_k)_{k \geq 0}$ in $\Q^n$ converging toward $\av{c}$ with
$\av{l}_k[i] \leq \av{c}[i] \leq \av{h}_k[i]$ for all $i$.
The set $C_k = \{\av{x} \in C \mid \bigwedge_i \av{l}_k[i] \leq \av{x}[i] \leq
\av{h}_k[i]\}$ is non-empty since $\av{c} \in C_k$.
By Lemma~\ref{lem:nonempty-rational}, there exists $\av{c}_k \in C_k \cap \Q^n$,
and the limit of $(\av{c}_k)_{k \geq 0}$ is $\av{c}$.
\end{proof}

\begin{theorem}
\label{thm:real}
Definable functions in $\FO(\R,+,\leq)$ are exactly the $\FO(\R,+,\leq)$-piecewise
linear functions.
\end{theorem}

\begin{proof}
Let $f : \R^n \to \R^q$ with $G_f$ definable in $\FO(\R,+,\leq)$.
We deduce that $G_f$ can be decomposed into a finite union of polyhedral convex sets.
From Lemma~\ref{lem:decomp}, we can assume that $G_f$ is a polyhedral convex set.
As $\dom(f) = \exists_{n+1,\ldots,n+q} G_f$, Fourier--Motzkin elimination shows that
$\dom(f)$ is polyhedral convex, hence definable in $\FO(\R,+,\leq)$.
If $\dom(f)$ is empty the proof is immediate.
Otherwise, from Lemma~\ref{lem:nonempty-rational} there exists
$\av{d}_0 \in \dom(f) \cap \Q^n$,
and for any $\av{d} \in \dom(f) \cap \Q^n$, since $(\av{d}, f(\av{d}))$ is the unique
vector of $G_f \cap (\{\av{d}\} \times \R^q)$, Lemma~\ref{lem:nonempty-rational} gives
$f(\av{d}) \in \Q^q$.

Let $D = \{(\av{x},z) \in \Q^n \times (\Qp \setminus \{0\}) \mid \av{d}_0 + \frac{\av{x}}{z} \in \dom(f)\}$
and $h : \R^n \times \R \to \R^q$ defined over $D$ by
$h(\av{x},z) = z(f(\av{d}_0 + \frac{\av{x}}{z}) - f(\av{d}_0))$.
The convexity of $G_f$ shows that $h$ is stable by addition.
Since $D \subseteq \Q^n \times \Qp$ and $f$ maps $\dom(f) \cap \Q^n$ into $\Q^q$
(established above), we have $G_h \subseteq \Q^{n+1} \times \Q^q$.
Proposition~\ref{prop:stable} then gives a matrix
$N \in \Q^{q \times (n+1)}$ such that $h(\av{x},z) = N(\av{x},z)$ for any
$(\av{x},z) \in \dom(h)$.

For any $\av{d} \in \dom(f) \cap \Q^n$, let $\av{x} = \av{d} - \av{d}_0$ and $z = 1$.
We get $f(\av{d}) = f(\av{d}_0) + N(\av{d} - \av{d}_0, 1)$, so there exist
$M \in \Q^{q \times n}$ and $\av{v} \in \Q^q$ with
$f(\av{d}) = M\av{d} + \av{v}$ for any $\av{d} \in \dom(f) \cap \Q^n$.

Finally, Proposition~\ref{prop:dense} shows that $G_f \cap (\Q^n \times \Q^q)$ is
dense in $G_f$.
There is a sequence $(\av{a}_i, f(\av{a}_i))_{i \geq 0}$ in $G_f \cap (\Q^n \times \Q^q)$
converging to $(\av{c}, f(\av{c}))$.
Since $\av{x} \mapsto M\av{x} + \av{v}$ is linear, it is continuous.
Therefore $f(\av{a}_i) = M\av{a}_i + \av{v} \to M\av{c} + \av{v}$, and thus $f(\av{c}) = M\av{c} + \av{v}$.
\end{proof}

\section{The Integer Additive Theory and the Ginsburg--Spanier Theorem}
\label{sec:integer}

\subsection{The Ginsburg--Spanier Theorem}

A set $L \subseteq \Z^n$ is said \emph{linear}~\cite{GS66} if there exist a vector
$\av{b} \in \Z^n$ and a finite set $P \subseteq \Z^n$ such that $L = \av{b} + P^*$
where $P^*$ is the \emph{submonoid} of $(\Z^n,+)$ generated by $P$.
Recall~\cite{GS66} that sets definable in $\FO(\Z,+,\leq)$ are exactly the finite unions
of linear sets (also called \emph{semi-linear sets}).

\begin{theorem}[\cite{GS66}]
\label{thm:GS}
The sets definable in $\FO(\Z,+,\leq)$ are exactly the semi-linear sets.
\end{theorem}

This characterization shows that the Presburger sets are exactly the \emph{rational sets}
of $(\Z^n,+)$.

\subsection{Presburger Functions Are Piecewise Linear}

\begin{theorem}
\label{thm:integer}
Definable functions in $\FO(\Z,+,\leq)$ are exactly the $\FO(\Z,+,\leq)$-piecewise linear
functions.
\end{theorem}

\begin{proof}
Let $f : \Z^n \to \Z^q$ such that $G_f$ is definable in $\FO(\Z,+,\leq)$.
Since $G_f$ is definable in $\FO(\Z,+,\leq)$, Theorem~\ref{thm:GS} gives a decomposition of $G_f$ into a finite union of linear sets.
From Lemma~\ref{lem:decomp}, we can assume that $G_f$ is a linear set.
Since $\dom(f) = \exists_{n+1,\ldots,n+q} G_f$ we deduce that $\dom(f)$ is definable in
$\FO(\Z,+,\leq)$.
As $G_f$ is a linear set, there exists $(\av{x}_0,\av{y}_0) \in \Z^n \times \Z^q$ and a
finite set $P \subseteq \Z^n \times \Z^q$ such that $G_f = (\av{x}_0,\av{y}_0) + P^*$.
We consider the function $h$ defined over
$\dom(h) = \dom(f) - \av{x}_0$ by
$h(\av{d}) = f(\av{x}_0 + \av{d}) - \av{y}_0$ for any $\av{d} \in \dom(h)$.
By definition of $h$, we have $G_h = G_f - (\av{x}_0,\av{y}_0) = P^*$,
so $h$ is stable by addition and $G_h \subseteq \Q^n \times \Q^q$.
Proposition~\ref{prop:stable} proves that there exists $M \in \Q^{q \times n}$ such that
$h(\av{d}) = M\av{d}$ for any $\av{d} \in \dom(h)$.
Now let $\av{x} \in \dom(f)$ and $\av{d} = \av{x} - \av{x}_0$.
From $h(\av{d}) = M\av{d}$ and $h(\av{d}) = f(\av{x}_0 + \av{d}) - \av{y}_0$ we deduce
$f(\av{x}) = M\av{x} + \av{v}$ where $\av{v} = \av{y}_0 - M\av{x}_0$.
Therefore $f$ is linear.
\end{proof}

\begin{remark}
In contrast to the operational characterizations of Presburger functions via counter
automata~\cite{MR67,Che76,GI79} and the composition-closure results of~\cite{IL81},
Theorem~\ref{thm:integer} provides a direct algebraic normal form
using only linear algebra and the Ginsburg--Spanier theorem.
\end{remark}

\begin{remark}
The piecewise linear characterization is implicitly present in the literature:
applying~\cite{GS66} to the graph $G_f \subseteq \Z^n \times \Z^q$ yields a semilinear
decomposition from which linearity on each piece follows by the above argument.
However, this consequence was never stated explicitly as a theorem with a direct proof.
\end{remark}

\section{Extension of the Ginsburg--Spanier Theorem: Semi-polinear Sets and Functions}
\label{sec:mixed}

\subsection{Mixed-linear Sets}

We fix a countable set $\mathcal{X}$ of variables.
Valuations are totally-defined functions $v : \mathcal{X} \to \R$ and we write
$v \models \phi$ if $v$ satisfies a formula $\phi$.
When $\phi$ is a formula of the mixed linear arithmetic, the denoted set $S \subseteq \R^n$
is said to be \emph{mixed-linear}.

In~\cite{Wei99}, Weispfenning showed that the mixed-linear arithmetic is decidable
via a quantifier elimination algorithm for an effectively equivalent extended logic.
The description of mixed-linear \emph{sets} in~\cite{Wei99} is incomplete;
Theorem~\ref{thm:mixed-sets} below provides the correct characterization via
semi-polinear sets (see Remark~\ref{rem:wei99}).

\subsection{Definition of Semi-polinear Sets}

Let us recall that a \emph{finitely generated monoid} is a subset $M \subseteq \Z^n$ such
that there exists a finite set $P \subseteq \Z^n$ satisfying $M = P^*$ where
$P^* := \{\av{0}\} \cup \{p_1 + \cdots + p_k \mid k \geq 1,\, p_i \in P\}$.

\begin{definition}
\label{def:polinear}
A \emph{polinear set} is a subset $L \subseteq \R^n$ of the form
\[
L := C + b + M \;=\; \{c + z \mid c \in C,\; z \in b + M\},
\]
where $b \in \Z^n$, $M \subseteq \Z^n$ is a finitely generated monoid, and
$C \subseteq [0,1)^n$ is a non-empty polyhedral convex set.
A \emph{semi-polinear set} is a finite union of polinear sets.
\end{definition}

\begin{remark}
The term \emph{polinear} combines \emph{polyhedral} (for the convex part $C$)
and \emph{linear} (for the integer part $b + M$).
\end{remark}

Since every vector in $\R^n$ decomposes \emph{uniquely} as $c + z$ with
$c \in [0,1)^n$ and $z \in \Z^n$, each element of $L$ has fractional part in $C$
and integer part in $b + M$.
When $C = \{\av{0}\}$, the polinear set reduces to the linear set $b + M \subseteq \Z^n$.
The linear sets of Definition~\ref{def:polinear} coincide with those of~\cite{GS66},
and semi-linear sets are a special case of semi-polinear sets.

\subsection{Closure Properties}

We denote by $e_i \in \Z^n$ the $i$-th standard basis vector,
and by $\exists_i S = \{\av{s} \in \R^{n-1} \mid \exists s_i\in\R:\,
(s_1,\ldots,s_n) \in S\}$ the \emph{$i$-th component removal}.

\begin{proposition}
\label{prop:closure}
The class of semi-polinear sets is stable by finite union, finite intersection,
complement, and component removal.
\end{proposition}

\begin{proof}
The unique decomposition $s = c + z$ with $c \in [0,1)^n$ and $z \in \Z^n$ is the key:
every set operation acts independently on the polyhedral part (in $[0,1)^n$)
and on the semilinear part (in $\Z^n$).
It suffices to treat a single polinear set $L = C + b + M$
(the general case follows by distributing over finite unions).
Let also $L' = C' + b' + M'$ be a second polinear set.

\smallskip
\noindent\textbf{Intersection.}
Since the fractional/integer split is unique:
\[
L \cap L' = (C \cap C') + \bigl((b + M) \cap (b' + M')\bigr).
\]
$C \cap C'$ is polyhedral convex (intersection of two polyhedral convex sets in
$[0,1)^n$); $(b+M) \cap (b'+M')$ is semilinear by Theorem~\ref{thm:GS}.
If either part is empty, so is $L \cap L'$.

\smallskip
\noindent\textbf{Complement.}
For $s = c + z \in \R^n$:
$s \notin C + b + M$ iff $c \notin C$ or $z \notin b + M$.
Hence:
\[
\R^n \setminus L
  \;=\; \bigl(([0,1)^n \setminus C) + \Z^n\bigr)
        \;\cup\;
        \bigl(C + (\Z^n \setminus (b + M))\bigr).
\]
Both terms are semi-polinear:
$[0,1)^n \setminus C$ is a finite union of polyhedral convex sets in $[0,1)^n$
(Boolean closure of polyhedral sets), and $\Z^n = \{\pm e_1,\ldots,\pm e_n\}^*$
is finitely generated, so the first term is semi-polinear;
$\Z^n \setminus (b+M)$ is semilinear by Theorem~\ref{thm:GS},
so the second term is semi-polinear.

\smallskip
\noindent\textbf{Component removal.}
Let a single polinear set $L = C + b + M$, then:
\[
\exists_i L = (\exists_i C) + (\exists_i b) + (\exists_i M).
\]
$\exists_i C$ is polyhedral convex in $[0,1)^{n-1}$ by Fourier--Motzkin elimination;
$\exists_i M$ is a finitely generated monoid (image of $M$ under coordinate projection).
So $\exists_i L$ is semi-polinear.
\end{proof}

\subsection{Mixed-linear Sets Are Semi-polinear}

We introduce the addition relation
$G_+ := \{(\alpha,\beta,\gamma) \in \R^3 \mid \alpha + \beta = \gamma\}$
and the first-order logic $\FO(\R,\Z,\R_{\geq 0},G_+)$, which is effectively equivalent
to the mixed-linear arithmetic.

\begin{theorem}
\label{thm:mixed-sets}
The class of semi-polinear sets and the class of mixed-linear sets are equal.
\end{theorem}

\begin{proof}
It is clear that every semi-polinear set is mixed-linear.
For the converse, let $\mathcal{F}$ be the set of formulas $\phi$ of
$\FO(\R,\Z,\R_{\geq 0},G_+)$ such that, for any tuple of free variables
$(x_1,\ldots,x_n)$ of $\phi$, the set $\{(s_1,\ldots,s_n) \in \R^n \mid \phi[x_i \mapsto s_i]\}$
is semi-polinear.
We prove that $\mathcal{F}$ contains all formulas by structural induction.

The key case is the predicate $G_+(x_i,x_j,x_k)$.
The corresponding set $S_+$ is semi-polinear via the decomposition
\[
S_+ = \bigcup_{\delta \in \{0,1\}}
\bigl\{c \in [0,1)^n \;\big|\; c_i + c_j = c_k + \delta\bigr\} + \delta e_k + M,
\]
where $M := \{z \in \Z^n \mid z_i + z_j = z_k\}$ is a finitely generated monoid
(it is a free abelian group of rank $n-1$, generated as a monoid by $n-1$ basis
vectors and their negatives).
Indeed, decomposing $s \in \R^n$ as $s = c + z$ with $c \in [0,1)^n$ and $z \in \Z^n$,
$s \in S_+$ iff $(c_i + z_i) + (c_j + z_j) = (c_k + z_k)$,
i.e.\ $z_k - z_i - z_j = c_i + c_j - c_k$.
Since the left-hand side is an integer and the right-hand side belongs to $(-1,2)$,
there exists $\delta \in \{0,1\}$ such that $z_k - z_i - z_j = \delta$ and
$c_i + c_j - c_k = \delta$.

The remaining base cases are:
$\{s \mid s_i \in \Z\} = \{c \in [0,1)^n \mid c_i = 0\} + \Z^n$ (polinear, since
$\Z^n$ is generated by $\{\pm e_1,\ldots,\pm e_n\}$); and
$\{s \mid s_i \geq 0\} = [0,1)^n + \{z \in \Z^n \mid z_i \geq 0\}$ (polinear, since
$\{z \in \Z^n \mid z_i \geq 0\}$ is generated by $\{e_i\} \cup \{\pm e_j \mid j \neq i\}$).
The stability properties of Proposition~\ref{prop:closure} then handle the Boolean
connectives ($\neg$, $\wedge$, $\vee$) and existential quantification ($\exists r$)
by induction.
\end{proof}

\begin{remark}
\label{rem:wei99}
Weispfenning~\cite{Wei99} claims that mixed-linear sets of $\R$ are finite unions of sets
of the form $J + p\Z$ where $J$ is an interval with rational bounds and $p \in \Z$.
This is incorrect: $\N$ is a mixed-linear set but is not a finite union of sets of this
form (any $J + p\Z$ with $p \neq 0$ contains arbitrarily negative integers, and $p = 0$
gives a bounded interval).
The error stems from describing the periodic component as a \emph{subgroup} ($p\Z$)
rather than a \emph{submonoid} ($p\N$).
The correct form is $I + p\N$: when the monoid $M \subseteq \N$ has only non-negative
generators, by Theorem~\ref{thm:mixed-sets} the semi-polinear sets of $\R$ with such
monoids are exactly finite unions of sets $I + p\N$ where $I \subseteq [0,1)$ is a
rational interval and $p \geq 0$
(by the Frobenius lemma~\cite{Syl84}: every finitely generated submonoid of $\N$ is
$\{0\}$ or of the form $B + d\N$ with $B \subseteq \N$ finite and $d \geq 1$).
More generally, monoids with mixed-sign generators (e.g.\ $M = \Z$, generated by
$\{1,-1\}$) yield sets such as $I + \Z$ that are not of this form but are still
semi-polinear.

\noindent\textit{Relation to Weispfenning's ultimately periodically simple sets.}
Weispfenning works with the extended logic $\FO(\R,\Z,+,\leq,[\cdot])$, which adds the
integer-part function $[\cdot]$ to achieve quantifier elimination, and characterizes
its definable sets as the \emph{ultimately periodically simple sets}~\cite{Wei99}.
Since $[x]$ is already definable in $\FO(\R,\Z,+,\leq)$ (via $[x] \leq x < [x]+1$
with $[x] \in \Z$), both logics define the same sets.
With the submonoid correction above (replacing $p\Z$ by $p\N$), the ultimately
periodically simple sets of~\cite{Wei99} coincide exactly with the semi-polinear sets
of Theorem~\ref{thm:mixed-sets}.
\end{remark}

\begin{remark}[Debt to Weispfenning~\cite{Wei99}]
\label{rem:wei99-debt}
Although~\cite{Wei99} contains an error in the description of mixed-linear sets
(Remark~\ref{rem:wei99}), it was the source of a key structural intuition:
every set definable in $\FO(\R,\Z,+,\leq)$ decomposes as a finite union of sets
$Z + D$ with $Z \subseteq \Z^n$ definable in $\FO(\Z,+,\leq)$ and
$D \subseteq [0,1)^n$ definable in $\FO(\R,+,\leq)$.
The present paper makes this decomposition a theorem rather than a lemma from~\cite{Wei99}:
Theorem~\ref{thm:mixed-sets} establishes it from first principles via the
semi-polinear characterization, and the $Z + D$ structure is then recovered as a
consequence (taking $Z = b + M$ and $D = C$).
\end{remark}

\subsection{Mixed Definable Functions Are Piecewise-Simple}

We call a function $f : \R^n \to \R^q$ \emph{mixed-linear} if there exist matrices
$M, N \in \Q^{q \times n}$ and $\av{v} \in \Q^q$ such that
$f(\av{a}) = M\av{a} + N[\av{a}] + \av{v}$ for all $\av{a} \in \dom(f)$.

\begin{definition}
\label{def:simple}
A function $f : \R^n \to \R$ partially-defined over a semi-polinear set $D \subseteq \R^n$
is said to be \emph{simple} if there exist $\alpha \in \Q$ and $a, b \in \Q^n$ such that
\[
f(c + z) = \alpha + \sum_{i=1}^n a_i c_i + b_i z_i
\]
for any $c \in [0,1)^n$, $z \in \Z^n$ with $c + z \in D$.
A function $f$ is said \emph{piecewise-simple} if $\dom(f)$ admits a finite partition
into semi-polinear sets $D_1,\ldots,D_k$ such that $f|_{D_j}$ is simple for each $j$.
\end{definition}

\begin{remark}
\label{rem:simple}
A simple function is linear in the fractional part $c \in [0,1)^n$ and linear in the
integer part $z \in \Z^n$, but with potentially different slopes.
The integer-part function $[\cdot] : \R \to \Z$ is simple (with $a = 0$, $b = 1$),
as is the fractional part $\{\cdot\}$ (with $a = 1$, $b = 0$); neither is piecewise linear.
Every simple function is mixed-linear (with $N = b - a$ in the formula
$f(\av{a}) = M\av{a} + N[\av{a}] + \av{v}$), so every piecewise-simple function is
definable in $\FO(\R,\Z,+,\leq)$.
\end{remark}

\begin{theorem}
\label{thm:mixed-functions}
A function $f : \R^n \to \R^q$ is definable in $\FO(\R,\Z,+,\leq)$ if and only if
it is piecewise-simple.
\end{theorem}

\begin{proof}
If $f$ is piecewise-simple, then its graph is a finite union of semi-polinear sets:
on each piece, the equation defining the output is mixed-linear in the integer and
fractional parts.  Hence $f$ is definable by Theorem~\ref{thm:mixed-sets}.

For the converse, first assume $q=1$ and let $f : \R^n \to \R$ be definable in
$\FO(\R,\Z,+,\leq)$.
By Theorem~\ref{thm:mixed-sets}, its graph is semi-polinear.  Using
Lemma~\ref{lem:decomp}, each polinear component is the graph of the restriction of
$f$ to a semi-polinear domain.  The finitely many domains can be refined, using
Proposition~\ref{prop:closure}, into a semi-polinear partition.  It is therefore
enough to treat one polinear component of the graph; write it as
$G_f = C + b + M$ where $C \subseteq [0,1)^{n+1}$ is polyhedral convex,
$b \in \Z^{n+1}$, and $M = P^*$ for some finite $P \subseteq \Z^{n+1}$.

Since $\mathbf{0} \in M$, we have $C + b \subseteq G_f$; since $G_f$ is a function graph,
$C$ is the graph of a function $g : \pi_x(C) \to [0,1)$
(where $\pi_x : \R^{n+1} \to \R^n$ projects onto the first $n$ coordinates).
Similarly, $M = P^*$ is the graph of a function $h : \pi_x(M) \to \Z$.
Vectors $s \in \R^n$ are uniquely decomposed into $c + z$ with $c \in [0,1)^n$ and
$z \in \Z^n$, so $f(c + z) = b_{n+1} + g(c) + h(z - (b_1,\ldots,b_n))$.

\textit{The function $g$.}
As $C$ is a polyhedral convex set, it is the set of solutions of a conjunction of
equalities and inequalities $u_1 c_1 + \cdots + u_{n+1} c_{n+1} \mathbin{\#} \alpha$.
If one constraint is an equality with $u_{n+1} \neq 0$, then $C$ is the graph of a
simple function.
Otherwise, for each inequality $I$ of the form
$\beta + a_1 c_1 + \cdots + a_n c_n \mathbin{\#} c_{n+1}$
appearing in the description of $C$, let $C_I$ be the set of vectors $c \in C$ satisfying
$\beta + a_1 c_1 + \cdots + a_n c_n = c_{n+1}$.
The restriction of $g$ to $\exists_{n+1} C_I$ is simple.
Every point in $\dom(g)$ satisfies at least one such inequality tightly.
Indeed, fix $(c_1,\ldots,c_n) \in \dom(g)$: the fiber
$\{c_{n+1} \mid (c_1,\ldots,c_n,c_{n+1}) \in C\}$
is a polyhedral convex subset of $\R$ that is a singleton (since $g$ is a function).
A singleton polyhedral convex set in $\R$ is the intersection of its lower and upper
bounds, so at least one bounding inequality is tight at $c_{n+1} = g(c_1,\ldots,c_n)$.
Thus $\dom(g)$ is covered by the finite union of the sets $\exists_{n+1} C_I$.
The covering can again be refined into a finite polyhedral, hence semi-polinear,
partition.  Hence $g$ is piecewise-simple.

\textit{The function $h$.}
We introduce the \emph{orthogonal} $S^\perp$ of a set $S \subseteq \Q^{n+1}$ as the set
of vectors $v \in \Q^{n+1}$ such that $s_1 v_1 + \cdots + s_{n+1} v_{n+1} = 0$ for any $s \in S$.
We recall from~\cite{Sch87} that $(S^\perp)^\perp$ is exactly the $\Q$-linear span of $S$.

Assume by contradiction that $P^\perp \subseteq \Q^n \times \{0\}$.
Every $v \in P^\perp$ has $v_{n+1} = 0$, so $v \cdot e_{n+1} = 0$ for all $v \in P^\perp$,
hence $e_{n+1} \in (P^\perp)^\perp = \mathrm{span}_\Q(P)$.
Write $e_{n+1} = \sum_j r_j p_j$ with $r_j \in \Q$.
Choose $d \geq 1$ with $dr_j \in \Z$ for all $j$ and write $dr_j = n_j^+ - n_j^-$ with $n_j^\pm \in \N$.
Set $p^+ = \sum_j n_j^+ p_j \in P^*$ and $p^- = \sum_j n_j^- p_j \in P^*$.
For $i = 1,\ldots,n$: $(e_{n+1})[i] = 0$ gives $\sum_j r_j p_j[i] = 0$,
hence $p^+[i] = \sum_j n_j^+ p_j[i] = \sum_j n_j^- p_j[i] = p^-[i]$.
But $p^+[n+1] - p^-[n+1] = d \cdot (e_{n+1})[n+1] = d \geq 1$.
Thus $p^+$ and $p^-$ both lie in $G_h = P^*$ with the same $\Z^n$-projection
but different $\Z$-outputs, contradicting $h$ being a function.

Hence there exists $v = (v_1,\ldots,v_n, v_{n+1}) \in P^\perp$ with $v_{n+1} \neq 0$.
For any $p = (\av{z}, y) \in P^*$: since $p$ is a sum of generators from $P$ and
$v \perp P$, we have $v \cdot p = 0$, giving
\[
h(\av{z}) = y = -\frac{v_1 z_1 + \cdots + v_n z_n}{v_{n+1}} = M_h \av{z},
\]
where $M_h = -v_{n+1}^{-1}(v_1,\ldots,v_n) \in \Q^{1\times n}$.
Thus $h$ is linear on $\dom(h)$ (simple with $a_i = 0$ for all $i$, since the domain is
in $\Z^n$ so fractional parts are zero).
Hence $f(c + z) = b_{n+1} + g(c) + h(z - (b_1,\ldots,b_n))$ is piecewise-simple.

It remains only to record effectivity.  Suppose that the graph is given as an
explicit finite union of polinear sets $C+b+P^*$.  The refinement of the domains
uses the effective Boolean operations on polyhedral convex sets and semilinear
sets from Proposition~\ref{prop:closure}.  For the real part $g$, the pieces are
obtained by inspecting the finitely many defining constraints of $C$: equalities
with non-zero output coefficient, or active inequalities
$\beta+a_1c_1+\cdots+a_nc_n=c_{n+1}$; their projections are computed by
Fourier--Motzkin elimination.  For the integer part $h$, one computes a vector
$v\in P^\perp$ with $v_{n+1}\neq0$ by Gaussian elimination; this gives the
linear expression for $h$ explicitly.  These operations are polynomial in the
size of the current polinear component and in the number of pieces actually
produced.  Thus the conversion from a semi-polinear graph representation to a
piecewise-simple representation is effective in output-sensitive time.

For $q>1$, apply the scalar case to each coordinate function $f_1,\ldots,f_q$.
Each coordinate yields a finite semi-polinear partition on which it is simple.
Taking a common refinement of these partitions, using Proposition~\ref{prop:closure},
gives a finite semi-polinear partition on which all coordinates are simple
simultaneously.  Hence $f$ is piecewise-simple as a vector-valued function.
\end{proof}

\begin{corollary}
\label{thm:mixed-pw}
A function is definable in $\FO(\R,\Z,+,\leq)$ if and only if it is
$\FO(\R,\Z,+,\leq)$-piecewise mixed-linear.
\end{corollary}

\begin{proof}
Every piecewise-simple function is mixed-linear (Remark~\ref{rem:simple}),
so the result follows from Theorem~\ref{thm:mixed-functions}.
\end{proof}

\section{Conclusion}

We have established a unified algebraic framework for sets and functions definable in
the three additive theories.
In $\FO(\Z,+,\leq)$: semi-linear sets (Theorem~\ref{thm:GS}) and piecewise linear
functions (Theorem~\ref{thm:integer}).
In $\FO(\R,+,\leq)$: polyhedral sets and piecewise linear functions
(Theorem~\ref{thm:real}).
In $\FO(\R,\Z,+,\leq)$: semi-polinear sets (Theorem~\ref{thm:mixed-sets}) and
piecewise-simple functions (Theorem~\ref{thm:mixed-functions}).

All proofs are purely algebraic, using only linear algebra and the Ginsburg--Spanier
theorem; no automata or quantifier elimination procedures are needed.
We also correct an error in Weispfenning~\cite{Wei99} regarding the structure of
mixed-linear sets of $\R$ (Remark~\ref{rem:wei99}).

\paragraph{Open problem.}
The algebraic characterization of $\FO(\Z,+,V_p)$-definable functions
(B\"uchi arithmetic~\cite{Sem77}) --- the analogue of Theorem~\ref{thm:integer}
for Presburger arithmetic extended with the $p$-adic valuation ---
appears to be open.

\section*{Acknowledgements}

This paper is a substantially revised version of an unpublished manuscript by the authors from 2008.
The authors used a large language model (Claude by Anthropic) as a writing and verification assistant during the revision process.

\bibliographystyle{alphaurl}
\bibliography{presburger_gs}

\end{document}